\documentclass[aps,superscriptaddress,11pt,twoside]{revtex4}

\usepackage{amsmath,latexsym,amssymb,verbatim,enumerate,graphicx}
\usepackage[all]{xy}

\usepackage{theorem}
\newtheorem{definition}{Definition}[section]

\newtheorem{lemma}[definition]{Lemma}

\newtheorem{theorem}[definition]{Theorem}

\def\squareforqed{\hbox{\rlap{$\sqcap$}$\sqcup$}}
\def\qed{\ifmmode\squareforqed\else{\unskip\nobreak\hfil
\penalty50\hskip1em\null\nobreak\hfil\squareforqed
\parfillskip=0pt\finalhyphendemerits=0\endgraf}\fi}
\def\endenv{\ifmmode\;\else{\unskip\nobreak\hfil
\penalty50\hskip1em\null\nobreak\hfil\;
\parfillskip=0pt\finalhyphendemerits=0\endgraf}\fi}
\newenvironment{proof}{\noindent \textbf{{Proof~} }}{\qed}

\mathchardef\ordinarycolon\mathcode`\:
\mathcode`\:=\string"8000
\def\vcentcolon{\mathrel{\mathop\ordinarycolon}}
\begingroup \catcode`\:=\active
  \lowercase{\endgroup
  \let :\vcentcolon
  }

\newcommand{\nc}{\newcommand}
\nc{\rnc}{\renewcommand} \nc{\beq}{\begin{equation}}
\nc{\eeq}{{\end{equation}}} \nc{\bea}{\begin{eqnarray}}
\nc{\eea}{\end{eqnarray}} \nc{\beqa}{\begin{eqnarray}}
\nc{\eeqa}{\end{eqnarray}} \nc{\lbar}[1]{\overline{#1}}
\nc{\bra}[1]{\langle#1|} \nc{\ket}[1]{|#1\rangle}
\nc{\ketbra}[2]{|#1\rangle\!\langle#2|}
\nc{\braket}[2]{\langle#1|#2\rangle} \nc{\proj}[1]{|
#1\rangle\!\langle #1 |} \nc{\avg}[1]{\langle#1\rangle}
\rnc{\max}{\operatorname{max}} \nc{\rank}{\operatorname{rank}}
\nc{\conv}{\operatorname{conv}}
\nc{\smfrac}[2]{\mbox{$\frac{#1}{#2}$}} \nc{\Tr}{\operatorname{Tr}}
\nc{\ox}{\otimes} \nc{\dg}{\dagger} \nc{\dn}{\downarrow}
\nc{\cA}{{\cal A}} \nc{\cB}{{\cal B}} \nc{\cC}{{\cal C}}
\nc{\cD}{{\cal D}} \nc{\cE}{{\cal E}} \nc{\cF}{{\cal F}}
\nc{\cG}{{\cal G}} \nc{\cH}{{\cal H}} \nc{\cI}{{\cal I}}
\nc{\cJ}{{\cal J}} \nc{\cK}{{\cal K}} \nc{\cL}{{\cal L}}
\nc{\cM}{{\cal M}} \nc{\cN}{{\cal N}} \nc{\cO}{{\cal O}}
\nc{\cP}{{\cal P}} \nc{\cR}{{\cal R}} \nc{\cS}{{\cal S}}
\nc{\cT}{{\cal T}} \nc{\cU}{{\cal U}} \nc{\cV}{{\cal V}}
\nc{\cX}{{\cal X}} \nc{\cW}{{\cal W}} \nc{\cZ}{{\cal Z}}
\nc{\csupp}{{\operatorname{csupp}}}
\nc{\qsupp}{{\operatorname{qsupp}}} \nc{\rar}{\rightarrow}
\nc{\lrar}{\longrightarrow} \nc{\poly}{\operatorname{poly}}
\nc{\polylog}{\operatorname{polylog}} \nc{\Lip}{\operatorname{Lip}}
\rnc{\1}{I}
\nc{\supp}{{\operatorname{supp}}}

\def\>{\rangle}
\def\<{\langle}

\def\a{\alpha}
\def\b{\beta}

\def\d{\delta}
\def\e{\epsilon}

\def\l{\lambda}

\def\r{\rho}
\def\s{\sigma}

\def\ph{\varphi}

\def\ps{\psi}

\def\Ph{\Phi}

\def\O{\Omega}

\nc{\glneq}{{\raisebox{0.6ex}{$>$}  \hspace*{-1.8ex} \raisebox{-0.6ex}{$<$}}}
\nc{\gleq}{{\raisebox{0.6ex}{$\geq$}\hspace*{-1.8ex} \raisebox{-0.6ex}{$\leq$}}}

\nc{\RR}{{{\mathbb R}}}
\nc{\CC}{{{\mathbb C}}}
\nc{\FF}{{{\mathbb F}}}
\nc{\HH}{{{\mathbb H}}}
\nc{\NN}{{{\mathbb N}}}
\nc{\ZZ}{{{\mathbb Z}}}
\nc{\PP}{{{\mathbb P}}}
\nc{\QQ}{{{\mathbb Q}}}
\nc{\UU}{{{\mathbb U}}}
\nc{\WW}{{{\mathbb W}}}
\nc{\EE}{{{\mathbb E}}}
\rnc{\SS}{{{\mathbb S}}}
\nc{\id}{{\operatorname{id}}}

\nc{\vholder}[1]{\rule{0pt}{#1}}

\nc{\ob}[1]{#1}

\def\beq{\begin {equation}}
\def\eeq{\end {equation}}

\nc{\eq}[1]{Eq.~(\ref{eq:#1})} \nc{\eqs}[2]{Eqs.~(\ref{eq:#1}) and
(\ref{eq:#2})}

\nc{\eqn}[1]{Eq.~(\ref{eqn:#1})}
\nc{\eqns}[2]{Eqs.~(\ref{eqn:#1}) and (\ref{eqn:#2})}

\nc{\Hm}{{H^{\min}}} \nc{\Hpm}{{H_p^{\min}}} \rnc{\log}{\ln}
\nc{\Wg}{\operatorname{Wg}}

\begin{document}

\title{{\Large Counterexamples to the maximal $\mathbf{p}$-norm\\
               multiplicativity conjecture for all $\mathbf{p > 1}$}}

\author{Patrick Hayden}
 \email{patrick@cs.mcgill.ca}
 \affiliation{
    School of Computer Science,
    McGill University,
    Montreal, Quebec, H3A 2A7, Canada
    }

\author{Andreas Winter}
 \email{a.j.winter@bris.ac.uk}
  \affiliation{
     Department of Mathematics,
    University of Bristol,
    University Walk, Bristol BS8 1TW, U.~K.
    }
  \affiliation{
     Centre for Quantum Technologies,
     National University of Singapore,
     2 Science Drive 3, Singapore 117542}

\date{30 July 2008}

\begin{abstract}
For all $p > 1$, we demonstrate the existence of quantum channels
with non-multiplicative maximal output $p$-norms. Equivalently, for
all $p
>1$, the minimum output R\'enyi entropy of order $p$ of a quantum
channel is not additive. The violations found are large; in all
cases, the minimum output R\'enyi entropy of order $p$ for a product
channel need not be significantly greater than the minimum output
entropy of its individual factors. Since $p=1$ corresponds to the
von Neumann entropy, these counterexamples demonstrate that if the
additivity conjecture of quantum information theory is true, it
cannot be proved as a consequence of any channel-independent
guarantee of maximal $p$-norm multiplicativity. We also show that a
class of channels previously studied in the context of approximate
encryption lead to counterexamples for all $p > 2$.
\end{abstract}

\maketitle

\parskip .75ex

%%%%%%%%%%%%%%%%%%%%%%%%%%%%%%%%%%%%%%%%%%%%%%%%%%%%%%%%%%%%%%%%%%%%%%%%

\section{Introduction} \label{sec:intro}

The oldest problem of quantum information theory is arguably to
determine the capacity of a quantum-mechanical communications
channel for carrying information, specifically \emph{classical}
bits of information. (Until the 1990's it would have been
unnecessary to add that additional qualification, but today the
field is equally concerned with other forms of information like
\emph{qubits} and \emph{ebits} that are fundamentally
quantum-mechanical.) The classical capacity problem long predates
the invention of quantum source coding~\cite{Schumacher95,JozsaS94}
and was of concern to the founders of information theory
themselves~\cite{P73}.
\begin{comment}
Indeed, in 1973, John Pierce ended his contribution to a 25 year
retrospective on information theory with this now famous
provocation~\cite{P73}:
\begin{quotation}
I think that I have never met a physicist who understood information
theory. I wish that physicists would stop talking about
reformulating information theory and would give us a general
expression for the capacity of a channel with quantum effects taken
into account rather than a number of special cases.
\end{quotation}
\end{comment}
The first major result on the problem came with the resolution of a
conjecture of Gordon's~\cite{G64} by Alexander Holevo in 1973, when
he published the first proof~\cite{H73} that the
maximum amount of information that can be extracted from an ensemble
of states $\s_i$ occurring with probabilities $p_i$ is bounded above
by
\begin{equation}
\chi(\{p_i,\s_i\})
 = H\Big( \sum_i p_i \s_i \Big) - \sum_i p_i H(\s_i),
\end{equation}
where $H(\s) = -\Tr \s \ln \s$ is the von Neumann entropy of the
density operator $\s$. For a quantum channel $\cN$, one can then
define the Holevo capacity
\begin{equation}
\chi(\cN) = \max_{\{p_i,\r_i\}} \chi( \{ p_i, \cN(\r_i) \} ),
\end{equation}
where the maximization is over all ensembles of input states.
Writing $C(\cN)$ for the classical capacity of the channel $\cN$,
this leads easily to an upper bound of
\begin{equation} \label{eqn:c.chi}
C(\cN) \leq \lim_{n \rar \infty} \frac{1}{n} \chi(\cN^{\ox n}).
\end{equation}
It then took more than two decades for further substantial progress
to be made on the problem, but in 1996, building on recent
advances~\cite{HausladenJSWW96}, Holevo~\cite{H98} and
Schumacher-Westmoreland~\cite{SW97} managed to show that the upper
bound in Eq.~(\ref{eqn:c.chi}) is actually achieved. This was a
resolution of sorts to the capacity problem, but the limit in the
equation makes it in practice extremely difficult to evaluate. If
the codewords used for data transmission are restricted such that
they are not entangled across multiple uses of the channel, however,
the resulting \emph{product state capacity} $C_{1\infty}(\cN)$ has
the simpler expression
\begin{equation}
C_{1\infty}(\cN) = \chi(\cN).
\end{equation}
The additivity conjecture for the Holevo capacity asserts that
for all channels $\cN_1$ and $\cN_2$,
\begin{equation}
\chi(\cN_1 \ox \cN_2) = \chi(\cN_1) + \chi(\cN_2).
\end{equation}
This would imply, in particular, that $C_{1\infty}(\cN) = C(\cN)$,
or that entangled codewords do not increase the classical capacity
of a quantum channel.

\begin{comment}
In the course of the development of quantum information theory,
questions also arose about the additivity of a number of other
important quantities. For a bipartite quantum state $\r^{AB}$, the
\emph{entanglement of formation}~\cite{BennettDSW96} is defined to
be
\begin{equation} \label{eqn:e.o.f}
E_f(\r^{AB}) = \min_{\{p_i,\ket{\ph_i}\}} \sum_i p_i H\Big( \Tr_B
\proj{\ph_i}^{AB} \Big),
\end{equation}
where the minimization is over all pure state decompositions $\sum_i
p_i \proj{\ph_i}$ of the density operator $\r$. The minimal rate at
which of singlets must be consumed in order to manufacture many
copies of $\r$ if the holders of the $A$ and $B$ systems are
restricted to exchanging bits and acting locally is the
\emph{entanglement cost}~\cite{HaydenHT01}, which satisfies
\begin{equation}
E_c(\r) = \lim_{n \rar \infty} \frac{1}{n} E_f(\r).
\end{equation}
The ``'' version of the additivity conjecture for the
entanglement of formation states that $E_c(\r) = E_f(\r)$ and there
is an analogous stronger version for pairs of states $\r^{AB}$ and
$\s^{AB}$. As one of two ``extremal'' measures of
entanglement~\cite{HorodeckiHH00}, the other being the entanglement
of distillation, the entanglement of formation plays a central role
in the theory of mixed state entanglement and a great deal of effort
has been devoted to evaluating
it~\cite{Wootters98,VW01,VidalDC02,MatsumotoY04}.
\end{comment}

In 2003, Peter Shor~\cite{Sh04}, building on several previously
established
connections~\cite{Pomeransky03,AudenaertB04,MatsumotoSW04},
demonstrated that the additivity of the Holevo capacity, the
additivity of the entanglement of
formation~\cite{BennettDSW96,HaydenHT01,VidalDC02,MatsumotoY04} and
the superadditivity of the entanglement of formation~\cite{VW01} are
all equivalent to another conjecture of Shor's which is particularly
simple to express mathematically, known as the \emph{minimum output
entropy conjecture}~\cite{KingR01}.  For a channel $\cN$, define
\begin{equation} \label{eqn:min.entropy}
\Hm(\cN) = \min_{\ket{\ph}} H( \cN( \ph ) ),
\end{equation}
where the minimization is over all pure input states $\ket{\ph}$.
The minimum output entropy conjecture asserts that for all channels
$\cN_1$ and $\cN_2$,
\begin{equation} \label{eqn:add}
\Hm(\cN_1 \ox \cN_2) = \Hm(\cN_1) + \Hm(\cN_2).
\end{equation}

There has been a great deal of previous work on these conjectures,
particularly inconclusive numerical searches for counterexamples,
necessarily in low dimension, at Caltech, IBM, in Braunschweig
(IMaPh) and Tokyo (ERATO)~\cite{OsawaN01}, as well as proofs of many
special cases. For example, the minimum output entropy conjecture
has been shown to hold if one of the channels is the identity
channel~\cite{AmosovHW00,AmosovH01}, a unital qubit
channel~\cite{K01}, a generalized depolarizing
channel~\cite{FujiwaraH02,K03} or an entanglement-breaking
channel~\cite{Holevo98a,K01b,Sh02}. In addition, the weak additivity
conjecture was confirmed for generalized dephasing
channels~\cite{DevetakS05}, the conjugates of all these
channels~\cite{KingMNR05} and some other special classes of
channels~\cite{Cortese04,MatsumotoY04,DattaHS04,Fukuda05}. Further
evidence for qubit channels was supplied in~\cite{KingR01}. This
list is by no means exhaustive. The reader is directed to Holevo's
reviews for a detailed account of the history of the additivity
problem~\cite{Holevo04,Holevo07}.

For the past several years, the most commonly used strategy for
proving these partial results has been to demonstrate the
multiplicativity of maximal $p$-norms of quantum channels for $p$
approaching 1~\cite{AmosovHW00}. For a quantum channel $\cN$ and $p
> 1$, define the maximal $p$-norm of $\cN$ to be
\begin{equation}
\nu_p(\cN) = \sup\Big\{ \big\| \cN(\r) \big\|_p \, ; \r \geq 0, \,
\Tr \r =1 \Big\}.
\end{equation}
In the equation, $\| \s \|_p = \big( \Tr |\s|^p \big)^{1/p}$. The
\emph{maximal $p$-norm multiplicativity
conjecture}~\cite{AmosovHW00} asserts that for all quantum channels
$\cN_1$ and $\cN_2$,
\begin{equation} \label{eqn:p.norm.conjecture}
\nu_p(\cN_1 \ox \cN_2) = \nu_p(\cN_1)\nu_p(\cN_2).
\end{equation}
This can be re-expressed in an equivalent form more convenient to us
using R\'enyi entropies. Define the R\'enyi entropy of order $p$ to
be
\begin{equation}
H_p(\rho) = \frac{1}{1-p} \log \Tr \rho^p
\end{equation}
for $p>0$, $p \neq 1$. Since $\lim_{p \downarrow 1} H_p(\rho) =
H(\r)$, we will also define $H_1(\r)$ to be $H(\r)$. All these
entropies have the property that they are 0 for pure states and
achieve their maximum value of the logarithm of the dimension on
maximally mixed states. Define the minimum output R\'enyi entropy
$\Hpm$ by substituting $H_p$ for $H$ in Eq.~(\ref{eqn:min.entropy}).
Since $\Hpm(\cN) = \frac{p}{1-p}\ln \nu_p(\cN)$,
Eq.~(\ref{eqn:p.norm.conjecture}) can then be written equivalently
as
\begin{equation} \label{eqn:p.add}
\Hpm(\cN_1 \ox \cN_2) = \Hpm(\cN_1) + \Hpm(\cN_2),
\end{equation}
in which form it is clear that the maximal $p$-norm multiplicativity
conjecture is a natural strengthening of the original minimum output
entropy conjecture (\ref{eqn:add}).

This conjecture spawned a significant literature of its own which we
will not attempt to summarize. Holevo's reviews are again an
excellent source~\cite{Holevo04,Holevo07}. Some more recent
important references
include~\cite{KingR04,KingR05,SerafiniEW05,GiovannettiL04,DevetakJKR06,Michalakis07}.
Unlike the von Neumann entropy case, however, some counterexamples
had already been found prior to this paper. Namely, Werner and
Holevo found a counterexample to Eq. (\ref{eqn:p.add}) for $p >
4.79$~\cite{WH02} that nonetheless doesn't violate the $p$-norm
multiplicativity conjecture for $1 < p
<2$~\cite{AF05,Datta04,GiovannettiLR05}.

Moreover, in 2007, Winter showed that a class of channels that had previously been studied in the context of approximate encryption provide counterexamples to the conjecture for all $p > 2$~\cite{W07}.
In light of these developments, the standing conjecture was that the
maximal $p$-norm multiplicativity held for $1 \leq p \leq 2$,
corresponding to the region in which the map $X \mapsto X^p$ is
operator convex~\cite{KingR04}. More conservatively, it was
conjectured to hold at least in an open interval $(1,1+\e)$, which
would be sufficient to imply the minimum output entropy conjecture.
On the contrary, shortly after Winter's discovery, Hayden showed
that the conjecture is false for all $1 < p < 2$~\cite{H07}.

The current paper merges and slightly strengthens \cite{W07} and
\cite{H07}. We begin in Section \ref{sec:andreas}, by presenting
Winter's counterexamples from \cite{W07}, which share some important
features with \cite{H07} but are simpler to analyze. Section
\ref{sec:counterexamples} then presents Hayden's counterexamples
from \cite{H07} with an improved analysis showing that they work for
all $p>1$, not just $1 < p < 2$.

In particular, given $p>1$, we show that there
exist channels $\cN_1$ and $\cN_2$ with output dimension $d$ such
that both $\Hpm(\cN_1)$ and $\Hpm(\cN_2)$ are equal to $\log d -
\cO(1)$ but $\Hpm( \cN_1 \ox \cN_2 ) = \log d + \cO(1)$, so
\begin{equation}
 \Hpm( \cN_1 ) + \Hpm( \cN_2 ) - \Hpm( \cN_1 \ox \cN_2 )
 =  \log d - \cO(1).
\end{equation}
Thus, one finds that the minimum output entropy of the product
channel need not be significantly larger than the minimum output
entropy of the individual factors. Since~\cite{AmosovHW00,K03}
\begin{equation} \label{eqn:king}
\Hpm( \cN_1 \ox \cN_2 ) \geq \Hpm( \cN_1 ) = \log d - \cO(1),
\end{equation}
these counterexamples are
essentially the strongest possible for all $p > 1$, up to
a constant additive term.
(Note that the dependence of $\Hpm$ on $p$ is absorbed here in the asymptotic notation.)

At $p=1$ itself, however, we see no evidence of a violation of the
additivity conjecture for the channels we study. Thus, the
conjecture stands and it is still an open question whether entangled
codewords can increase the classical capacity of a quantum channel.

\begin{comment}
{\bf Structure of the paper:} The counterexamples to the maximal
$p$-norm multiplicativity conjecture are presented in section
\ref{sec:counterexamples}. The case of the von Neumann entropy,
$p=1$, is discussed in section \ref{sec:vn.entropy}. An appendix
describes a lengthy calculation helpful for understanding what is
happening at $p=1$.
\end{comment}

{\bf Notation:} If $A$ and $B$ are finite dimensional Hilbert
spaces, we write $AB \equiv A\otimes B$ for their tensor product and
$|A|$ for $\dim A$. The Hilbert spaces on which linear operators act
will be denoted by a superscript.  For instance, we write $\ph^{AB}$
for a density operator on $AB$. Partial traces will be abbreviated
by omitting superscripts, such as $\ph^A \equiv \Tr_B\ph^{AB}$.  We
use a similar notation for pure states, e.g.\ $\ket{\psi}^{AB}\in
AB$, while abbreviating $\psi^{AB} \equiv \proj{\psi}^{AB}$. We
associate to any two isomorphic Hilbert spaces $A\simeq A'$ a
unique maximally entangled state which we denote
$\ket{\Phi}^{AA'}$.  Given any orthonormal basis $\{\ket{i}^A\}$ for
$A$,  if we define $\ket{i}^{A'} = V\ket{i}^A$ where $V$ is the
associated isomorphism, we  can write this state as $\ket{\Ph}^{AA'}
= |A|^{-1/2}\sum_{i=1}^{|A|}\ket{i}^A\ket{i}^{A'}$.
We will also make use of the asymptotic notation $f(n) = \cO(g(n))$
if there exists $C > 0$ such that for sufficiently large $n$,
$|f(n)| \leq C g(n)$. $f(n) = \Omega(g(n))$ is defined similarly but
with the reverse inequality $|f(n)| \geq C g(n)$. Finally, $f(n) =
\Theta(g(n))$ if $f(n) = \cO(g(n))$ and $f(n) = \O(g(n))$.

\section{Random unitary channels: $\mathbf{p > 2}$}
\label{sec:andreas}

This class of counterexamples, while only working for $p
> 2$, has the advantage of being a straightforward
application of well-known results. Later in the paper we will
present stronger counterexamples that reuse the same basic strategy,
albeit with some additional technical complications.
A random unitary channel is a map of the form
\begin{equation}
  \label{eq:rand}
  {\cal N}:\rho \longmapsto \frac{1}{n} \sum_{i=1}^n V_i \rho V_i^\dagger,
\end{equation}
with the $V_i$ unitary transformations of an underlying (finite dimensional) Hilbert space. Let $d$ be the dimension
of this space.
%(More generally, one
%could allow variable probability weights for the different $V_i$,
%but we won't need that here.)
%
Following~\cite{HLSW04}, we call ${\cal N}$
\emph{$\epsilon$-randomizing} if for all $\rho$,
\begin{equation}
  \label{eq:opnorm-epsrand}
 % \forall\rho\qquad
 \left\| {\cal N}(\rho) - \frac{1}{d}\1 \right\|_\infty
                                                   \leq \frac{\epsilon}{d}.
\end{equation}
In that paper, it was shown that for $0<\epsilon<1$,
$\epsilon$-randomizing channels exist in all dimensions $d >
\frac{10}{\epsilon}$, with $n = \frac{134}{\epsilon^2} d \log d$. In
fact, randomly picking the $V_i$ from the Haar measure on the
unitary group will, with high probability, yield such a channel.

Recently, it was shown by Aubrun~\cite{Aubrun:no-log} that $n$ can
in fact be taken to be $\cO(d/\epsilon^2)$ for Haar distributed $V_i$,
and $\cO(d(\log d)^4/\epsilon^2)$ for $V_i$ drawn from any ensemble
of exactly randomizing unitaries.

\begin{lemma} \label{lem:1}
For a random unitary channel ${\cal N}$ and its
complex conjugate, $\overline{\cal N}:\rho \mapsto \frac{1}{n} \sum
\overline{V_i} \rho \overline{V_i}^\dagger$, one has $\nu_p({\cal
N}\ox\overline{\cal N}) \geq \frac{1}{n}$.
\end{lemma}
\begin{proof}
We use the maximally entangled state
$\ket{\Phi} = d^{-1/2}\sum_{i} \ket{i}\ket{i}$ as test state, abbreviating $\Phi = \proj{\Phi}$:
\[\begin{split}
  \nu_p({\cal N}\ox\overline{\cal N})
      &\geq \left\| ({\cal N}\ox\overline{\cal N})\Phi \right\|_p \\
      &=    \left\| \frac{1}{n^2}\sum_{i,j=1}^n (V_i\ox\overline{V_j})
                                                 \Phi
                                               (V_i\ox\overline{V_j})^\dagger \right\|_p \\
      &=    \left\| \frac{1}{n}\Phi
                     + \frac{1}{n^2}\sum_{i\neq j} (V_i\ox\overline{V_j})
                                                     \Phi
                                                   (V_i\ox\overline{V_j})^\dagger \right\|_p
       \geq \frac{1}{n},
\end{split}\]
where in the third line we have invoked the
$U\ox\overline{U}$-invariance of $\Phi$ for the $n$ terms when $i=j$. For the final inequality, observe that the largest
eigenvalue $\lambda_1$ of $({\cal N}\ox\overline{\cal
N})\Phi$ is at least $\frac{1}{n}$. Denoting the other eigenvalues
$\lambda_{\a}$, $\|  ({\cal N}\ox\overline{\cal
N})\Phi\|_p = \left( \sum_{\a} \lambda_{\a}^p
\right)^{1/p} \geq \lambda_1$, and we are done.
\end{proof}

\begin{lemma} \label{lem:2}
If the channel ${\cal N}$ is $\epsilon$-randomizing, then when $p>1$,
\[
  \nu_p({\cal N}) = \nu_p(\overline{\cal N}) \leq \left( \frac{1+\epsilon}{d} \right)^{1-1/p}.
\]
\end{lemma}
\begin{proof}
Clearly, ${\cal N}$ and $\overline{\cal N}$ have the same maximum
output $p$-norm. For the former, observe that the
$\epsilon$-randomizing condition implies that for an arbitrary input state
$\rho$, $\| {\cal N}(\rho) \|_\infty \leq \frac{1+\epsilon}{d}$. In
other words, all the eigenvalues $\lambda_{\a}$ of the output state
${\cal N}(\rho)$ are bounded between $0$ and $\frac{1+\epsilon}{d}$. In addition,
because ${\cal N}(\rho)$ is a density operator, the eigenvalues sum to $1$.

Subject to these constraints, however, the convexity of the function $x\mapsto x^p$
ensures that the $p$-norm
$\| {\cal N}(\rho) \|_p = \left( \sum_{\a} \lambda_{\a}^p \right)^{1/p}$
is maximized when the
largest eigenvalue is $\frac{1+\epsilon}{d}$ and it occurs with multiplicity
$\lfloor \frac{d}{1+\epsilon} \rfloor$, and all but possibly one remaining eigenvalue is $0$.
Thus,
\begin{equation}
\| {\cal N}(\rho) \|_p =    \left( \sum_{\a} \lambda_{\a}^p \right)^{1/p}
                        \leq \left( \frac{d}{1+\epsilon}
                                    \left( \frac{1+\epsilon}{d} \right)^p \right)^{1/p}
                        =    \left( \frac{1+\epsilon}{d} \right)^{1-1/p}.
\end{equation}
\end{proof}
\begin{theorem}
Fix any $0<\epsilon<1$ and a family of $\epsilon$-randomizing maps
${\cal N}$ as in Eq. (\ref{eq:rand}) with $n > 134 \, d \ln d /
\e^2$. Then, for any $p>2$ and sufficiently large $d$,
\begin{equation}
  \label{eq:main}
  \nu_p({\cal N}) \nu_p(\overline{\cal N})
                      \leq \left( \frac{1+\epsilon}{d} \right)^{2-2/p}
                      \ll  \frac{1}{n}
                      \leq \nu_p({\cal N}\ox\overline{\cal N}),
\end{equation}
In other words,
for this family of channels, the maximum output $p$-norm is strictly
supermultiplicative for sufficiently large $d$ when $p>2$.
\end{theorem}
\begin{proof}
Follows from Lemmas \ref{lem:1} and \ref{lem:2} since $2-2/p > 1$.
\end{proof}

These counterexamples to the multiplicativity of the output $p$-norm
for $p>2$ are interesting in that they are random unitary channels,
which are among the simplest truly quantum maps. In fact, the first
proofs of multiplicativity for unital qubit channels~\cite{K01} and
depolarizing channels~\cite{K03} exploited this type of structure.
Indeed, unital qubit channels are always random unitary channels
(with $d=2$)~\cite{KingR01}. Despite the fact that King showed
multiplicativity for such channels at all $p>1$~\cite{K01}, there is
no conflict with the result here, as the bound on $n$ becomes better
than $d^2$ only for rather large dimension $d$.

We observe, furthermore, that $p=2$ is indeed the limit of validity
of this class of counterexamples, since $n \geq d$ for any
$\epsilon$-randomizing map.

\section{Generic quantum channels: all $\mathbf{p>1}$}
\label{sec:counterexamples}
Let $E$, $F$ and $G$ be finite dimensional quantum systems, then
define $R=E$, $S=FG$, $A=EF$ and $B=G$, so that $RS = AB = EFG$.
\begin{comment}
\begin{equation}
 \begin{array}{r}
 R \Big\{ E  \\
 S \left\{ \begin{array}{c} F \\ G \end{array} \right.
 \end{array}
 \begin{array}{l}
 \left. \begin{array}{l}  \\  \end{array} \right\} \\
 \}
 \end{array}
\end{equation}
\end{comment}
Our second and stronger class of counterexamples will be channels from $S$ to $A$ of the form
\begin{equation} \label{eqn:the.channel}
\cN( \r ) = \Tr_B \big[ U ( \proj{0}^R \ox \r ) U^\dg \big]
\end{equation}
for $U$ unitary and $\ket{0}$ some fixed state on $R$.
Another, slightly more flexible way of writing this is in the
language of isometric Stinespring dilations: namely, the Hilbert
space isometry $V:S\hookrightarrow AB$ defined by
$V\ket{\ph} = U\bigl( \ket{0}^R\ket{\ph}^S \bigr)$. In this
notation, to which we will adhere from now, $\cN(\rho) = \Tr_B V\rho V^\dagger$.

Our method will be to fix the dimensions of the systems involved,
select $U$ (i.e., the isometry $V$) at
random, and show that the resulting channel is likely to violate
additivity. The rough intuition motivating our examples is the same as in the previous section:
we will
exploit the fact that there are channels that appear to be highly
depolarizing for product state inputs despite the fact that they are
not close to the depolarizing channel in, for example, the norm of
complete boundedness~\cite{Paulsen86}.

Consider a single copy of $\cN$ and the associated map
$V: \ket{\ph}^S \mapsto U \bigl(\ket{0}^R \ket{\ph}^S\bigr)$.
This map takes $S$ to
a subspace of $A \ox B$, and if $U$ is selected according to the
Haar measure, then the image of $S$ is itself a random subspace,
distributed according to the unitarily invariant measure. In
\cite{HLW06}, it was shown that if $|S|$ is chosen appropriately,
then the image is likely to contain only almost maximally entangled
states, as measured by the entropy of entanglement. After tracing
over $B$, this entropy of entanglement becomes the entropy of the
output state. Thus, for $S$ of suitable size, all input states get
mapped to high entropy output states. We will repeat the analysis
below, finding that the maximum allowable size of $S$ will depend on
$p$ as described by the following two lemmas.
\begin{lemma}
  \label{lemma:Hp-expectation-Lipschitz}
  The maps $f_p(\ket{\ph}) = H_p(\ph^A)$ on unit vectors (states)
  $\ket{\ph} \in A \ox B$, $2\leq |A| \leq |B|$, have expectation
  \begin{equation}
    \label{eq:Hp-expectation}
    \EE f_p \geq \EE f_\infty \geq \ln |A| - \gamma \sqrt{|A| / |B|},
  \end{equation}
  for a uniformly random state $\ph$,
  with a universal constant $\gamma$ which may be chosen
  arbitrarily close to $3$ for sufficiently large $|A|$.
  \\
  Furthermore, for $p > 1$, the functions $f_p$ are all Lipschitz continuous, with the Lipschitz
  constant $\Lambda_p$ bounded above by
  \begin{equation}
    \label{eq:Hp-Lipschitz}
    \Lambda_p^2 \leq \frac{4 p^2}{(1-p)^2} |A|^{1-\frac{1}{p}}.
  \end{equation}
\end{lemma}
\begin{proof}
  The first inequality in Eq.~(\ref{eq:Hp-expectation}) is by the
  monotonicity of the R\'{e}nyi entropies in $p$. For the second,
  observe $f_\infty(\ket{\ph}) = -\ln \| \ph^A \|_\infty$, so
  \[
    \EE f_\infty(\ket{\ph}) =    \EE \bigl( -\ln \left\| \ph^A \right\|_\infty \bigr)
                            \geq -\ln \EE \left\| \ph^A \right\|_\infty.
  \]
  The expectation of the largest eigenvalue of $\ph^A$ has been widely studied
  in random matrix theory. Just note that $\ket{\ph}$ is
  well-approximated by a Gaussian unit vector $\ket{\Gamma}$,
  that is, a random vector all of whose real and imaginary
  components (in any basis) are i.i.d.~normal with expectation
  $0$ and variance $1/2|A||B|$. (See~\cite[Appendix]{BHLSW05}.)
  Indeed, by the triangle inequality,
  \[
    \EE_\ph \left\| \ph^A \right\|_\infty \leq \EE_\Gamma \left\| \Gamma^A \right\|_\infty,
  \]
  and the right hand side, for large $A$ and $B$, is
  known~\cite{Geman1980,Johnstone2001} to be asymptotically
  \[
    \frac{\left(\sqrt{|A|}+\sqrt{|B|}\right)^2}{|A||B|}
         =    \frac{1}{|A|} + \frac{2}{\sqrt{|A||B|}} + \frac{1}{|B|}
         \leq \frac{1}{|A|} \left( 1 + 3\sqrt{\frac{|A|}{|B|}} \right).
  \]
  The explicit upper bound of $\frac{\left(\sqrt{|A|}+\sqrt{|B|}\right)^2}{|A||B|}$
  has been obtained for matrices with \emph{real} Gaussian
  entries~\cite{DavidsonSzarek2001}, but the analogous statement for
  complex Gaussian entries seems to be unknown.

  Now, for the Lipschitz bound: we proceed as in~\cite{HLW06}, inferring the general
  bound from a Lipschitz bound for the R\'{e}nyi entropy of a dephased
  version on $\ph^A$.
  Fix bases $\{ \ket{j} \}$ and $\{ \ket{k} \}$ of $A$ and $B$, respectively,
  so that we can write
  $\ket{\ph} = \sum_{jk} \ph_{jk}\ket{j}\ket{k}$, where the coefficients
  are to be decomposed into real and imaginary parts: $\ph_{jk} = t_{jk0} + i t_{jk1}$.

  We actually show that
  \[
    g_p(\ket{\ph}) = \frac{1}{1-p} \ln \Tr \left[ \left( \sum_j \proj{j} \ph^A \proj{j} \right)^p \right]
            = \frac{1}{1-p} \ln \sum_j \bra{j} \ph^A \ket{j}^p
            = \frac{1}{1-p} \ln \sum_j \left( \sum_{kz} t_{jkz}^2 \right)^p
  \]
  is $\frac{2p}{p-1} |A|^{1/2-1/2p}$-Lipschitz. This implies the result for $f_p$ as follows.
  Note first that $g_p(\ket{\ph}) \geq f_p(\ket{\ph})$, with equality if $\{ \ket{j} \}$
  is an eigenbasis of $\ph^A$. Now,
  for two vectors $\ket{\ph}$, $\ket{\psi}$, we may without loss of
  generality assume that $f_p(\ket{\psi}) \geq f_p(\ket{\ph})$, and that $\{ \ket{j} \}$
  is the eigenbasis of $\ph^A$. Thus, by assumption,
  \[
    f_p(\ket{\psi})-f_p(\ket{\ph}) \leq g_p(\ket{\psi})-g_p(\ket{\ph})
                                   \leq \frac{2p}{p-1} |A|^{1/2-1/2p} \| \ket{\psi}-\ket{\ph} \|_2.
  \]

  To bound the Lipschitz constant of $g_p$, it is sufficient to
  find an upper bound on its gradient. It is straightforward to see that
  \[
    \frac{\partial g_p}{\partial t_{jkz}}
            = \frac{1}{1-p} \frac{1}{\sum_{j'} \left( \sum_{k'z'} t_{j'k'z'}^2 \right)^p}
                                           \cdot 2\,p\, t_{jkz} \left( \sum_{k'z'} t_{jk'z'}^2 \right)^{p-1},
  \]
  so introducing the notation $x_j = \sum_{kz} t_{jkz}^2$, we have
  \[
    \bigl\| \nabla g_p \bigr\|_2^2 = \frac{4p^2}{(1-p)^2}
                                     \frac{\sum_j x_j^{2p-1}}{\left( \sum_j x_j^p \right)^2}
                                 = \frac{4p^2}{(1-p)^2}
                                   \frac{\sum_j (x_j^p)^{(2p-1)/p}}{\left( \sum_j x_j^p \right)^2},
  \]
  which we need to maximize subject to the constraint $\sum_j x_j =1$.
  Since $(2p-1)/p \geq 1$, the function $y^{(2p-1)/p}$ is convex. Therefore, for  fixed $s = \sum_j x_j^p \geq |A|^{1-p}$,  the right hand side is maximal
  when all the $x_j^p$ except for one are $0$. Thus,
  \[
    \bigl\| \nabla g_p \bigr\|_2^2
          \leq \max_{|A|^{1-p} \leq s \leq 1} \frac{4p^2}{(1-p)^2} s^{[(2p-1)/p]-2}
          =    \frac{4p^2}{(1-p)^2} |A|^{1-1/p},
  \]
  and we are done.
\end{proof}

\begin{lemma} \label{lem:size.subspace}
Let $A$ and $B$ be quantum systems with $2 \leq |A| \leq |B|$ and
$1 < p \leq \infty$.
Then there exists a subspace $S \subset A \ox B$ of dimension
\begin{equation}
 |S| =  \left\lfloor
            \frac{c}{4} \left(1-\frac{1}{p}\right)^2 \frac{\alpha^2}{\ln(5/\delta)} |A|^{1/p}|B|
        \right\rfloor
\end{equation}
(with a universal constant $c$), that contains only states
$\ket{\ph} \in S$ with high entanglement, in the sense that
\begin{equation}
  \label{eqn:good.entropy}
  H_p(\ph^A) \geq \log |A| - \a - \b + \ln(1-\delta),
\end{equation}
where $\b = \gamma \sqrt{|A|/|B|}$ is as in Lemma~\ref{lemma:Hp-expectation-Lipschitz}.
The probability that a subspace of
dimension $|S|$ chosen at random according to the unitarily
invariant measure will not have this property is bounded above by
\begin{equation}
  2 \left( \frac{5}{\delta} \right)^{2|S|}
     \exp \left( -c \left( 1-\frac{1}{p} \right)^2 \alpha^2 |A|^{1/p}|B| \right).
\end{equation}
The universal constant $c$ may be chosen to be $1/72\pi^3$.
%Moreover, for $p = 1$, the probability that a random subspace does not obey (\ref{eqn:good.entropy}) is bounded above by
%\begin{equation} \label{eqn:subspace.prob}
%\left( \frac{10 \sqrt{1 + (\ln |A|)^2}}{\a} \right)^{2|S|}
%    \exp \left( - \frac{(2|A||B|-1)\a^2}{8(1 + (\ln |A|)^2)} \right).
%\end{equation}
\end{lemma}
\begin{proof}
The argument is nearly identical to the proof of Theorem IV.1
in~\cite{HLW06}, but with an improvement, possible due to the fact
the we are looking at a function defined via a norm. (See~\cite{L01}
and~\cite{Aubrun:no-log}.)

First of all, by Levy's Lemma, for a function $f$ on pure states of
$A\ox B$ with Lipschitz constant $\Lambda$, the random variable $f(\ket{\ph})$
for a uniformly distributed $\ket{\ph}$ on the unit sphere
in $A \ox B$ obeys
\[
  \Pr\{ f < \EE f - \alpha \} \leq
      2 \exp\left( -\frac{2}{9\pi^3} \frac{\alpha^2}{\Lambda^2} |A||B| \right).
\]
(See~\cite[Lemma III.1]{HLW06} for an exposition.) We apply this to $f_p$,
for which we have a Lipschitz bound by Lemma~\ref{lemma:Hp-expectation-Lipschitz}.
Furthermore, we can find a $\delta$-net ${\cal M}$ of cardinality
$|{\cal M}| \leq \left( 5/\delta \right)^{2|S|}$
on the unit vectors in $S$~\cite[Lemma III.6]{HLW06}.
In other words, for each unit vector $\ket{\ph} \in S$ there exists a $\ket{\tilde\ph} \in {\cal M}$
such that $\| \ket{\ph}-\ket{\tilde\ph} \|_2 \leq \delta$.
Combining the net, the Lipschitz constant and the union bound, we get
\[
  \Pr_S \Bigl\{ \exists \ket{\ph}\in {\cal M}
                          \quad f_p(\ket{\ph}) < \ln |A| - \alpha/2 - \beta \Bigr\}
    \leq \left( \frac{5}{\delta} \right)^{2|S|}
            2 \exp\left( -\frac{2}{9\pi^3} \frac{\alpha^2(1-1/p)^2}{16|A|^{1-1/p}} |A||B| \right),
\]
which is the probability inequality claimed in the theorem.
Moreover, the right hand side is less than $1$ if $|S|$ is chosen as
stated in the theorem.

Now, assume we have a subspace $S$ with a $\delta$-net ${\cal M}$ such that
\begin{equation}
  \label{eq:maximum-p-norm}
  (\forall \ket{\ph} \in {\cal M}) \left(  f_p(\ph) \geq \log |A| - \alpha -\beta\right), \ \text{ i.e. }\
  r := \max_{\ket{\ph}\in{\cal M}} \| \ph^A \|_p \leq e^{-(1-1/p)(\log|A|-\a-\b)}.
\end{equation}
Denote
\[
  R := \max_{\ket{\ph}\in S \text{ unit vector}} \| \ph^A \|_p
     = \max_{\rho \text{ d.o. supported on }S} \| \rho^A \|_p,
\]
where the latter equality is due to the convexity of the norm.
Hence, for each unit vector $\ket{\ph} \in S$ and corresponding
$\ket{\tilde\ph} \in {\cal M}$ such that $\| \ph - \tilde\ph \|_1 \leq \delta$,
\[
  \| \ph^A \|_p \leq \| \tilde\ph^A \|_p + \| \ph^A-\tilde\ph^A \|_p
                \leq r + \delta R,
\]
where we have used triangle inequality and the trace norm bound on
$\ph - \tilde\ph$. Consequently, $R \leq r/(1-\delta)$, and
inserting that into Eq.~(\ref{eq:maximum-p-norm}) finishes the
proof.
\end{proof}

Consider now the product channel $\cN \ox \bar{\cN}$, where
$\bar\cN(\r) = \Tr_B \bar V \r V^T$ is the complex conjugate of $\cN$.
We will exploit an approximate version of the symmetry used in the
random unitary channel counterexamples. Fix orthonormal bases of $S$, $A$ and $B$
to be used in the definition of maximally entangled states involving these systems.
(These have to be the same product bases with respect to which we define
the complex conjugate.)

In the trivial case where $|S| = |A \ox B|$, the isometry $V$ is
unitary and the identity
$V \ox \bar{V} \ket{\Ph} = (V \bar{V}^T \ox I) \ket{\Ph} = \ket{\Ph}$
for the maximally entangled state $\ket{\Ph}^{S_1 S_2}$ implies that
\begin{equation}
(\cN \ox \bar{\cN})(\proj{\Ph}^{S_1 S_2})
 = \Tr_{B_1 B_2} \left[ \proj{\Ph}^{A_1 A_2} \ox \proj{\Ph}^{B_1 B_2} \right]
 = \proj{\Ph}^{A_1 A_2}.
\end{equation}
The output of $\cN \ox \bar{\cN}$ will thus be a pure state. In the
general case, we will choose $|S|/|A\ox B|$ to be large but not trivial, in
which case useful bounds can still be placed on the largest
eigenvalue of the output state for an input state maximally
entangled between $S_1$ and $S_2$.
\begin{lemma}
  \label{lem:big.eigenvalue}
  Let $\ket{\Ph}^{S_1 S_2}$ be a state maximally entangled between
  $S_1$ and $S_2$ as in the previous paragraph. Then $(\cN \ox
  \bar\cN)(\Ph^{S_1S_2})$ has an eigenvalue of at least
  $\smfrac{|S|}{|A||B|}$.
\end{lemma}
\begin{proof}
This is an easy calculation again exploiting the $U\ox\bar U$
invariance of the maximally entangled state.
Note that whereas $V$ is an isometric embedding, $V^\dagger$ is a
partial isometry. More precisely, it can be understood as a unitary
$U^\dagger$ on $A \ox B$ followed by a fixed projection $P$, say
onto the first $|S|$ coordinates of $A \ox B$. Now,
\begin{equation*}\begin{split}
  \left\| (\cN \ox \bar{\cN})\proj{\Ph}^{S_1 S_2} \right\|_\infty
       &\geq \Tr\left( \bigl[ (\cN \ox \bar{\cN})\proj{\Ph}^{S_1 S_2} \bigr] \proj{\Ph}^{A_1 A_2} \right) \\
       &\geq \Tr\left( (V \ox \bar{V})\proj{\Ph}^{S_1 S_2}(V \ox \bar{V})^\dagger
                                                   (\proj{\Ph}^{A_1 A_2} \ox \proj{\Ph}^{B_1 B_2}) \right) \\
       &\!\!\!\!\!\!\!\!\!\!\!\!\!\!\!\!\!\!
        =    \Tr\left( (P \ox \bar{P})\proj{\Ph}^{S_1 S_2}(P \ox \bar{P})(U \ox \bar{U})^\dagger
                                   (\proj{\Ph}^{A_1 A_2} \ox \proj{\Ph}^{B_1 B_2}) (U\ox \bar{U})) \right) \\
       &\!\!\!\!\!\!\!\!\!\!\!\!\!\!\!\!\!\!
        =    \Tr\left( (P \ox \bar{P})\proj{\Ph}^{S_1 S_2}(P \ox \bar{P})
                                                  (\proj{\Ph}^{A_1 A_2} \ox \proj{\Ph}^{B_1 B_2}) \right)
        =    \frac{|S|}{|A||B|},
\end{split}\end{equation*}
and we are done.
\end{proof}

In order to demonstrate violations of additivity, the first step is
to bound the minimum output entropy from below for a single copy of
the channel. Fix $1 < p \leq \infty$, let $|B| = |A|$ so that $\b = \gamma$,
set $\a = \d = 1/2$, and then choose $|S|$ according to Lemma
\ref{lem:size.subspace}. With probability approaching $1$ as $|A| \rar
\infty$,
\begin{equation} \label{eqn:Hpm.1copy}
\Hpm( \cN ) \geq \ln |A| - \gamma - 1/2 - \ln 2,
\end{equation}
when the subspace $S$ defining the channel is chosen according to
the unitary invariant measure. (Since we're interested in $|A| \rar
\infty$, we may choose any $\gamma > 3$.) The same obviously holds
for $\Hpm( \bar\cN )$. Recall that the entropy of the uniform
distribution is $\ln |A|$ so the minimum entropy is near the maximum
possible. Fix a channel such that these lower bounds on $\Hpm( \cN
)$ and $\Hpm( \bar\cN )$ are satisfied.

By Lemma \ref{lem:big.eigenvalue},
\begin{equation}
 \label{eqn:renyi.bound}
  H_p\big( ( \cN \ox \bar\cN)(\Ph) )
     = \frac{1}{1-p} \ln \left( \sum_\a \l_a^p \right)
     \leq \frac{1}{1-p} \ln \left(\frac{|S|}{|A||B|}\right)^p
    = \frac{p}{1-p} \ln \frac{|S|}{|A||B|},
\end{equation}
where the $\l_\a$ are the eigenvalues of $( \cN \ox \bar\cN)(\Ph)$.
Substituting the value of $|S|$ from Lemma~\ref{lem:size.subspace}
into this inequality yields
\begin{equation}
  \label{eqn:Hpm.2copies}
  H_p\big( (\cN \ox \bar\cN)(\Ph) \big)
    \leq \log |A| + \cO\left( 1 + \frac{p}{p-1}\log\frac{p}{p-1} \right)
    \leq \log |A| + \cO\left( (1-1/p)^{-2} \right),
\end{equation}
where the $\cO$ notation hides only an absolute constant,
independent of $|A|$ and $p>1$. Thus, the R\'enyi entropy of $(\cN
\ox \bar\cN)(\Ph)$ is strictly less than $\Hpm(\cN) + \Hpm(\bar\cN)
\geq 2 \ln |A| - \cO(1)$. This is a violation of conjecture
(\ref{eqn:p.add}), with the size of the gap approaching $\ln |A| -
\cO(1) $ for large $|A|$.

%Since a theorem seems to be in order:
\begin{theorem} \label{thm:p.counterexamples}
For all $1 < p \leq \infty$, there exists a quantum channel for which the
inequalities (\ref{eqn:Hpm.1copy}) and (\ref{eqn:Hpm.2copies}) both
hold. The inequalities are inconsistent with the maximal $p$-norm
multiplicativity conjecture.
%(\ref{eqn:p.add}).
\qed
\end{theorem}

Note, however, that changing $p$ also requires changing $|S|$
according Lemma \ref{lem:size.subspace}, so we have a sequence of
channels violating additivity of the minimal output R\'enyi entropy
as $p$ decreases to 1, as opposed to a single channel doing so for
every $p$. This prevents us from drawing conclusions about the von
Neumann entropy by taking the limit $p \rar 1$. Likewise, an
examination of Eq.~(\ref{eqn:Hpm.2copies}) reveals that we also lose
control over the two-copy minimum output entropy of a fixed channel
as $p \rar 1$.

Another observation comes from the fact that our examples violate
additivity by so much: namely that, due to
Lemma~\ref{lem:big.eigenvalue}, the dimension of the subspace $S$ in
Lemma~\ref{lem:size.subspace} is essentially optimal up to constant
factors (depending on $p$). Any stronger violations of additivity
would contradict Eq.~(\ref{eqn:king}), the inequality $\Hpm(\cN \ox
\bar{\cN}) \geq \Hpm(\cN)$.

As an aside, it is interesting to observe that violating maximal
$p$-norm multiplicativity has structural consequences for the
channels themselves. For example, because entanglement-breaking
channels do not violate multiplicativity~\cite{King03b}, there must
be states $\ket{\ps}^{S_1 S_2}$ such that $(\cN \ox I^{S_2})(\ps)$
is entangled, despite the fact that $\cN$ will be a rather noisy
channel. (The same conclusions apply to the maps of section \ref{sec:andreas} ,
where the conclusion takes the form
that $\e$-randomizing  random unitary channels need not be entanglement-breaking.)

\section{The von Neumann entropy case}
\label{sec:vn.entropy} Despite the large violations found for $p$
close to 1, the class of examples presented here do not appear to
contradict the minimum output entropy conjecture for the von Neumann
entropy. The reason is that the upper bound demonstrated for
$H_p\big((\cN\ox\bar\cN)(\Ph)\big)$ in the previous section rested
entirely on the existence of one large eigenvalue for
$(\cN\ox\bar\cN)(\Ph)$. The von Neumann entropy is not as sensitive
to the value of a single eigenvalue as are the R\'{e}nyi entropies
for $p > 1$ and, consequently, does not appear to exhibit additivity
violations. With a bit of work, it is possible to make these
observations more rigorous.
\begin{lemma} \label{lem:avg.purity}
Let $\ket{\Ph}^{S_1 S_2}$ be a maximally entangled state between
$S_1$ and $S_2$. Assuming that $|A| \leq |B| \leq |S|$,
\begin{equation} \label{eqn:avg.purity}
 \int \Tr\Big[ \big( (\cN \ox \bar{\cN})(\proj{\Ph}) \big)^2 \Big] \, dU
    = \frac{|S|^2}{|A|^2|B|^2} + \cO\left(\frac{1}{|A|^2}\right),
\end{equation}
where ``$dU$'' is the normalized Haar measure on $R \ox S \cong A
\ox B$.
\end{lemma}
\begin{figure}
    \centering
    \includegraphics[scale=0.85]{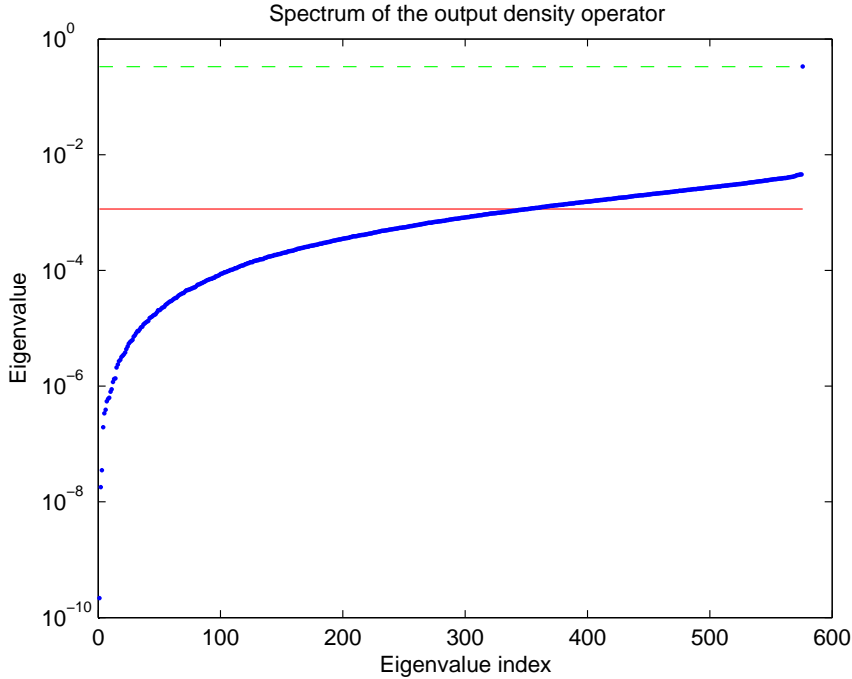}
    % plot_ortho_eig(3,8,24);Ø
    \caption{Typical eigenvalue spectrum of $(\cN \ox \bar\cN)(\Ph)$ when $|R|=3$ and
    $|A|=|B|=24$. The eigenvalues are plotted in increasing order
    from left to right. The green dashed line corresponds to
    $|S|/(|A||B|) = 1/3$, which is essentially equal to the largest
    eigenvalue. The red solid line represents the value
    $(1 - \smfrac{|S|}{|A||B|}) / |A|^2 = 1/864$. If the density operator
    were maximally mixed aside from its largest eigenvalue, all but
    that one eigenvalue would fall on this line. While that is not
    the case here or in general, the remaining eigenvalues are nonetheless sufficiently
    small to ensure that the density operator has high von Neumann
    entropy.}
    \label{fig:eigenspectrum}
\end{figure}
A description of the calculation can be found in Appendix
\ref{appendix:purity}. Let the eigenvalues of $(\cN\ox\bar\cN)(\Ph)$
be equal to $\l_1 \geq \l_2 \geq \cdots \geq \l_{|A|^2}$. For a
typical $U$, Lemmas \ref{lem:big.eigenvalue} and
\ref{lem:avg.purity} together imply that
\begin{equation}
\sum_{j > 1} \l_j^2 = \cO\left(\frac{1}{|A|^2}\right).
\end{equation}
Thus, aside from $\l_1$, the eigenvalues $\l_j$ must be quite small.
A typical eigenvalue distribution is plotted in Figure
\ref{fig:eigenspectrum}. If we define $\tilde\l_j = \l_j / (1 -
\l_1)$, then $\sum_{j>1} \tilde\l_j = 1$ and
\begin{equation}
H_1(\tilde\l)
 \geq H_2(\tilde\l)
 = - \ln \sum_{j>1} \tilde\l_j^2
 = 2 \ln |A| - \cO(1).
\end{equation}
An application of the grouping property then gives us a good lower
bound on the von Neumann entropy:
\begin{equation}
H_1\big((\cN\ox\bar\cN)(\Ph)\big)
 = H_1(\l)
 = h( \l_1 ) + (1 - \l_1) H_1 ( \tilde\l )
 = 2 \ln |A| - \cO(1),
\end{equation}
where $h$ is the binary entropy function. This entropy is nearly as
large as it can be and, in particular, as large as $\Hm( \cN ) +
\Hm(\bar\cN)$ according to Theorem IV.1 of \cite{HLW06}, the von
Neumann entropy version of Lemma \ref{lem:size.subspace}.

\section{Discussion} \label{sec:discussion}

The counterexamples presented here demonstrate that the maximal
$p$-norm multiplicativity conjecture and, equivalently, the minimum
output $p$-R\'{e}nyi entropy conjecture are false for all $1 < p
\leq \infty$. The primary motivation for studying this conjecture
was that it is a natural strengthening of the minimum (von Neumann)
output entropy conjecture, which is of fundamental importance in
quantum information theory. In particular, since the
multiplicativity conjecture was formulated, most attempts to prove
the minimum output entropy conjecture for special cases actually
proved maximal $p$-norm multiplicativity and then took the limit as
$p$ decreases to 1. This strategy, we now know, cannot be used to
prove the conjecture in general.
%; or multiplicativity might hold only up to a
%correction factor which however would tend to $1$ fast enough as $p \rar 1$.

From that perspective, it would seem that the results in this paper
cast doubt on the validity of the minimum output entropy conjecture
itself. However, as we have shown, the examples explored here appear
to be completely consistent with the conjecture, precisely because
the von Neumann entropy is more difficult to perturb than the
R\'{e}nyi entropies of order $p > 1$. It is therefore still possible
that the $p=1$ conjecture could be demonstrated using subtle
variants of $p$-norm multiplicativity such as exact or approximate
multiplicativity in a channel-dependent interval $(1,1+\d)$.

Another strategy that is still open would be to approach the von
Neumann minimum output entropy via R\'{e}nyi entropies for $p<1$. It
is possible that additivity holds there even as it fails for $p>1$.
That is not, unfortunately, a very well-informed speculation. With
few exceptions~\cite{WolfE05}, there has been very little research
on the additivity question in the regime $p < 1$, even though many
arguments can be easily adapted to this parameter region.
(Eq.~(\ref{eqn:king}), for example, holds for all $0 < p$.)
Unfortunately, since the time the examples presented here were first
circulated, counterexamples for $p$ close to $0$ were also
discovered~\cite{p-close-to-0}, casting doubt on the conjecture for
the whole set of R\'{e}nyi entropies with $p<1$. Indeed, as in the
current paper, those examples are based on influencing a single
eigenvalue of the output state of the tensor product channel; while
here we increase the largest one, there the smallest is suppressed.

%\medskip
Thus, while it seems doubtful that the examples of channels presented
here will have direct implications for the addivity of the minimum von Neumann entropy,
we think that they are still very useful
as a new class of test cases. Indeed, as we remarked earlier,
our examples eliminate what had been the previously favoured route to the conjecture via the output $p$-norms.

%\medskip
As a final comment, while this paper has demonstrated that the
maximal $p$-norm additivity conjecture fails for $p>1$, all the
counterexamples presented here have been nonconstructive. For the
examples based on $\epsilon$-randomizing maps, all the known
explicit constructions (by Ambainis and Smith~\cite{AmbainisSmith}
or via iterated quantum expander maps~\cite{q-expanders,Hastings})
only give bounds in the $2$-norm, which do imply bounds on the
output $p$-norm but those are too weak to yield counterexamples to
multiplicativity. Likewise, the counterexamples based on generic
quantum channels rely on the existence of large subspaces containing
only highly entangled states. Even when the entanglement is
quantified using von Neumann entropy, in which case the existence of
these subspaces was demonstrated in 2003~\cite{HLW06}, not a single
explicit construction is known. The culprit, as in many other
related contexts~\cite{PWPVJ08}, is our use of the probabilistic
method. Since we don't have any \emph{explicit} counterexamples,
only a proof that counterexamples exist, it remains an open problem
to ``derandomize'' our argument.

\subsection*{Acknowledgments}
We would like to thank Fr\'ed\'eric Dupuis and Debbie Leung for an
inspiring late-night conversation at the Perimeter Institute, Aram
Harrow for several insightful suggestions, and Mary Beth Ruskai for
discussions on the additivity conjecture. We also thank BIRS for
their hospitality during the Operator Structures in Quantum
Information workshop, which rekindled our interest in the additivity
problem.
PH was supported by the Canada Research Chairs program, a
Sloan Research Fellowship, CIFAR, FQRNT, MITACS, NSERC and QuantumWorks.
AW received support from the U.K.~EPSRC, the Royal Society and the
European Commission (project ``QAP'').

\appendix

\section{Proof of Lemma \ref{lem:avg.purity}} \label{appendix:purity}

We will estimate the integral, in what is perhaps not the most
illuminating way, by expressing it in terms of the matrix entries of
$U$. Let $U_{s,ab} = \,^{R} \bra{0}^{S}\bra{s} U \ket{a}^A
\ket{b}^B$. Expanding gives
\begin{align} \label{eqn:big.sum}
\int \Tr\Big[ & \big( (\cN \ox \bar{\cN})(\proj{\Ph}) \big)^2 \Big] \, dU \\
 &= \frac{1}{|S|^2}
    \sum_{\stackrel{a_1,a_2}{a_1',a_2'}}
    \sum_{\stackrel{b_1,b_2}{b_1',b_2'}}
    \sum_{\stackrel{s_1,s_2}{s_1',s_2'}}
    \int
    \bar U_{s_1,a_2 b_2} \bar U_{s_2,a_1' b_1}
            \bar U_{s_1',a_2' b_2'} \bar U_{s_2',a_1 b_1'}
            U_{s_1,a_1 b_1} U_{s_2,a_2' b_2}
            U_{s_1',a_1' b_1'} U_{s_2',a_2 b_2'} \, dU. \nonumber
\end{align}
Following \cite{AL03,AL04}, the non-zero terms in the sum can be
represented using a simple graphical notation. Make two parallel
columns of four dots, then label the left-hand dots by the indices
$(s_1, s_2, s_1', s_2')$ and the right-hand dots by the indices
$\vec v = ( a_2 b_2, a_1'b_1, a_2'b_2',a_1b_1')$. Join dots with a
solid line if the corresponding $\bar{U}$ matrix entry appears in
Eq.~(\ref{eqn:big.sum}). Since terms integrate to a non-zero value
only if the vector of $U$ indices $\vec w =
(a_1b_1,a_2'b_2,a_1'b_1',a_2b_2')$ is a permutation of the vector of
$\bar U$ indices, a non-zero integral can be represented by using a
dotted line to connect left-hand and right-hand dots whenever the
corresponding $U$ matrix entry appears in the integral.

Assuming for the moment that the vertex labels in the left column
are all distinct and likewise for the right column, the integral
evaluates to the Weingarten function $\Wg(\pi)$, where $\pi$ is the
permutation such that $w_i = v_{\pi(i)}$. For the rough estimate
required here, it is sufficient to know that $\Wg(\pi) = \Theta\big(
(|A||B|)^{-4-|\pi|} \big)$, where $|\pi|$ is the minimal number of
factors required to write $\pi$ as a product of transpositions, and
that $\Wg(e) = (|A||B|)^{-4}\big( 1 +
\cO(|A|^{-2}|B|^{-2})\big)$~\cite{CS06}.

The dominant contribution to Eq.~(\ref{eqn:big.sum}) comes from the
``stack'' diagram
\[
\xy
 (0,0)*{\bullet}="l4";  (0,7)*{\bullet}="l3";
 (0,14)*{\bullet}="l2"; (0,21)*{\bullet}="l1";
 (13,0)*{\bullet}="r4"; (13,7)*{\bullet}="r3";
 (13,14)*{\bullet}="r2";(13,21)*{\bullet}="r1";
 "l1" ; "r1" **\dir{-}; "l2" ; "r2" **\dir{-};
 "l3" ; "r3" **\dir{-}; "l4" ; "r4" **\dir{-};
 (-4,21)*{s_1};  (-4,14)*{s_2};
 (-4,7)*{s_1'}; (-4,0)*{s_2'};
 (25,21)*{a_2 b_2 = a_1 b_1};
 (25,14)*{a_1' b_1 = a_2' b_2};
 (25,7)*{a_2' b_2' = a_1' b_1'};
 (25,0)*{a_1 b_1' = a_2 b_2',};
\endxy
\]
in which the solid and dashed lines are parallel and for which the
contribution is positive and approximately equal to
\begin{equation} \label{eqn:dominant}
 \frac{1}{|S|^2}
    \sum_{\stackrel{a_1,a_2}{a_1',a_2'}}
    \sum_{\stackrel{b_1,b_2}{b_1',b_2'}}
    \sum_{\stackrel{s_1,s_2}{s_1',s_2'}}
    \d_{a_1a_2}\d_{b_1b_2}\d_{a_1'a_2'}\d_{b_1'b_2'}
    \Wg(\id)
% = \frac{1}{|S|^2} |A|^2 |B|^2
%    \Theta\left(\frac{1}{|A|^4|B|^4}\right)
 = \frac{|S|^2}{|A|^2|B|^2}\left( 1 + \cO(|A|^{-2}|B|^{-2} ) \right).
\end{equation}
(The expression on the left-hand side would be exact but for the
terms in which vertex labels are not distinct.) To obtain an
estimate of Eq.~(\ref{eqn:big.sum}), it is then sufficient to
examine the other terms and confirm that they are all of smaller
asymptotic order than this. There are six diagrams representing
transpositions, and their associated (negative) contributions are
\[
\xy
 (0,0)*{\xy
     (0,0)*{\bullet}="l4";  (0,7)*{\bullet}="l3";
     (0,14)*{\bullet}="l2"; (0,21)*{\bullet}="l1";
     (13,0)*{\bullet}="r4"; (13,7)*{\bullet}="r3";
     (13,14)*{\bullet}="r2";(13,21)*{\bullet}="r1";
     "l1" ; "r1" **\dir{-}; "l2" ; "r2" **\dir{-};
     "l3" ; "r3" **\dir{-}; "l4" ; "r4" **\dir{-};
     (-4,21)*{s_1};  (-4,14)*{s_2};
     (-4,7)*{s_1'}; (-4,0)*{s_2'};
     (25,21)*{a_2 b_2 = a_1 b_1};
     (25,14)*{a_1' b_1 = a_2' b_2};
     (25,7)*{a_2' b_2' = a_2 b_2'};
     (25,0)*{a_1 b_1' = a_1' b_1'};
     "l3"; "r4" **\dir{--};
     "l4"; "r3" **\dir{--};
     (15,-7)*{\Theta( |S|^2|A|^{-4}|B|^{-2})}
     \endxy};
 (50,0)*{\xy
     (0,0)*{\bullet}="l4";  (0,7)*{\bullet}="l3";
     (0,14)*{\bullet}="l2"; (0,21)*{\bullet}="l1";
     (13,0)*{\bullet}="r4"; (13,7)*{\bullet}="r3";
     (13,14)*{\bullet}="r2";(13,21)*{\bullet}="r1";
     "l1" ; "r1" **\dir{-}; "l2" ; "r2" **\dir{-};
     "l3" ; "r3" **\dir{-}; "l4" ; "r4" **\dir{-};
     (-4,21)*{s_1};  (-4,14)*{s_2};
     (-4,7)*{s_1'}; (-4,0)*{s_2'};
     (25,21)*{a_2 b_2 = a_1 b_1};
     (25,14)*{a_1' b_1 = a_2 b_2' };
     (25,7)*{a_2' b_2' = a_1' b_1'};
     (25,0)*{a_1 b_1' = a_2' b_2};
     "l2"; "r4" **\dir{--};
     "l4"; "r2" **\dir{--};
     (15,-7)*{\Theta( |S|^2|A|^{-4}|B|^{-4})}
     \endxy};
 (100,0)*{\xy
     (0,0)*{\bullet}="l4";  (0,7)*{\bullet}="l3";
     (0,14)*{\bullet}="l2"; (0,21)*{\bullet}="l1";
     (13,0)*{\bullet}="r4"; (13,7)*{\bullet}="r3";
     (13,14)*{\bullet}="r2";(13,21)*{\bullet}="r1";
     "l1" ; "r1" **\dir{-}; "l2" ; "r2" **\dir{-};
     "l3" ; "r3" **\dir{-}; "l4" ; "r4" **\dir{-};
     (-4,21)*{s_1};  (-4,14)*{s_2};
     (-4,7)*{s_1'}; (-4,0)*{s_2'};
     (25,21)*{a_2 b_2 = a_2 b_2'};
     (25,14)*{a_1' b_1 = a_2' b_2};
     (25,7)*{a_2' b_2' = a_1' b_1'};
     (25,0)*{a_1 b_1' = a_1 b_1,};
     "l1"; "r4" **\dir{--};
     "l4"; "r1" **\dir{--};
     (15,-7)*{\Theta( |S|^2|A|^{-2}|B|^{-4})}
     \endxy}
\endxy
\]
\[
\xy
 (0,0)*{\xy
     (0,0)*{\bullet}="l4";  (0,7)*{\bullet}="l3";
     (0,14)*{\bullet}="l2"; (0,21)*{\bullet}="l1";
     (13,0)*{\bullet}="r4"; (13,7)*{\bullet}="r3";
     (13,14)*{\bullet}="r2";(13,21)*{\bullet}="r1";
     "l1" ; "r1" **\dir{-}; "l2" ; "r2" **\dir{-};
     "l3" ; "r3" **\dir{-}; "l4" ; "r4" **\dir{-};
     (-4,21)*{s_1};  (-4,14)*{s_2};
     (-4,7)*{s_1'}; (-4,0)*{s_2'};
     (25,21)*{a_2 b_2 = a_1 b_1};
     (25,14)*{a_1' b_1 = a_1' b_1'};
     (25,7)*{a_2' b_2' = a_2' b_2};
     (25,0)*{a_1 b_1' = a_2 b_2'};
     "l3"; "r2" **\dir{--};
     "l2"; "r3" **\dir{--};
     (15,-7)*{\Theta( |S|^2|A|^{-2}|B|^{-4})}
     \endxy};
 (50,0)*{\xy
     (0,0)*{\bullet}="l4";  (0,7)*{\bullet}="l3";
     (0,14)*{\bullet}="l2"; (0,21)*{\bullet}="l1";
     (13,0)*{\bullet}="r4"; (13,7)*{\bullet}="r3";
     (13,14)*{\bullet}="r2";(13,21)*{\bullet}="r1";
     "l1" ; "r1" **\dir{-}; "l2" ; "r2" **\dir{-};
     "l3" ; "r3" **\dir{-}; "l4" ; "r4" **\dir{-};
     (-4,21)*{s_1};  (-4,14)*{s_2};
     (-4,7)*{s_1'}; (-4,0)*{s_2'};
     (25,21)*{a_2 b_2 = a_1' b_1'};
     (25,14)*{a_1' b_1 = a_2' b_2 };
     (25,7)*{a_2' b_2' = a_1 b_1};
     (25,0)*{a_1 b_1' = a_2 b_2'};
     "l1"; "r3" **\dir{--};
     "l3"; "r1" **\dir{--};
     (15,-7)*{\Theta( |S|^2|A|^{-4}|B|^{-4})}
     \endxy};
 (100,0)*{\xy
     (0,0)*{\bullet}="l4";  (0,7)*{\bullet}="l3";
     (0,14)*{\bullet}="l2"; (0,21)*{\bullet}="l1";
     (13,0)*{\bullet}="r4"; (13,7)*{\bullet}="r3";
     (13,14)*{\bullet}="r2";(13,21)*{\bullet}="r1";
     "l1" ; "r1" **\dir{-}; "l2" ; "r2" **\dir{-};
     "l3" ; "r3" **\dir{-}; "l4" ; "r4" **\dir{-};
     (-4,21)*{s_1};  (-4,14)*{s_2};
     (-4,7)*{s_1'}; (-4,0)*{s_2'};
     (25,21)*{a_2 b_2 = a_2' b_2};
     (25,14)*{a_1' b_1 = a_1 b_1};
     (25,7)*{a_2' b_2' = a_1' b_1'};
     (25,0)*{a_1 b_1' = a_2 b_2'};
     "l1"; "r2" **\dir{--};
     "l2"; "r1" **\dir{--};
     (15,-7)*{\Theta( |S|^2|A|^{-4}|B|^{-2}).}
     \endxy}
\endxy
\]
For permutations $\pi$ such that $|\pi| > 1$, the Weingarten
function is significantly suppressed: $\Wg(\pi) = \cO( |A|^{-6}
|B|^{-6})$. Moreover, for a given diagram type, the requirement that
$w_i = v_{\pi(i)}$ can only hold if at least two pairs of the
indices $a_1,a_2,a_1',a_2',b_1,b_2,b_1',b_2'$ are identical. The
contribution from such diagrams is therefore
$\cO(|S|^2|A|^{-4}|B|^{-2})$.

To finish the proof, it is necessary to consider integrals in which
the vertex labels on the left- or the right-hand side of a diagram
are not all distinct. In this more general case, choosing a set
$\cC$ of representatives for the conjugacy classes of the
permutation group on four elements, the value of the integral can be
written
\begin{equation} \label{eqn:cycle.decomp}
 \sum_{c \in \cC} N(c) \Wg(c),
\end{equation}
where
\begin{equation}
 N(c) = \sum_{\stackrel{\s \in \cS_4:}{\vec v = \s(\vec v)}}
        \sum_{\stackrel{\tau \in \cS_4:}{\vec w = \tau(\vec w)}}
        \d( \tau \pi \s \in c ).
\end{equation}
These formulas have a simple interpretation. Symmetry in the vertex
labels introduces ambiguities in the diagrammatic notation; the
formula states that every one of the diagrams consistent with a
given vertex label set must be counted, and with a defined
dimension-independent multiplicity. Conveniently, our crude
estimates have already done exactly that, ignoring the
multiplicities. The only case for which we need to know the
multiplicities, moreover, is for contributions to the dominant term,
which we want to know exactly and not just up to a constant
multiple.

We claim that in the sum (\ref{eqn:big.sum}) there are at most
$\cO(|S|^4|A||B|^3)$ terms with vertex label symmetry. The total
contribution for terms with vertex label symmetries $\tau$ and $\s$
in which $|\tau \pi \s| \geq 1$ is therefore of size
$\cO(|S|^2|A|^{-4}|B|^{-2})$ and does not affect the dominant term.
To see why the claim holds, fix a diagram type and recall that the
requirement $w_i = v_{\pi(i)}$ for a permutation $\pi$ can only hold
if at least two pairs of the indices
$a_1,a_2,a_1',a_2',b_1,b_2,b_1',b_2'$ are identical. Equality is
achieved only when all the $A$ indices or all the $B$ indices are
aligned, corresponding to the following two diagrams:
\[
\xy
 (0,0)*{\xy
     (0,0)*{\bullet}="l4";  (0,7)*{\bullet}="l3";
     (0,14)*{\bullet}="l2"; (0,21)*{\bullet}="l1";
     (13,0)*{\bullet}="r4"; (13,7)*{\bullet}="r3";
     (13,14)*{\bullet}="r2";(13,21)*{\bullet}="r1";
     "l1" ; "r1" **\dir{-}; "l2" ; "r2" **\dir{-};
     "l3" ; "r3" **\dir{-}; "l4" ; "r4" **\dir{-};
     (-4,21)*{s_1};  (-4,14)*{s_2};
     (-4,7)*{s_1'}; (-4,0)*{s_2'};
     (25,21)*{a_2 b_2 = a_2' b_2};
     (25,14)*{a_1' b_1 = a_1 b_1};
     (25,7)*{a_2' b_2' = a_2 b_2'};
     (25,0)*{a_1 b_1' = a_1' b_1'};
     "l1"; "r2" **\dir{--};
     "l2"; "r1" **\dir{--};
     "l3"; "r4" **\dir{--};
     "l4"; "r3" **\dir{--};
     \endxy};
 (50,0)*{\xy
     (0,0)*{\bullet}="l4";  (0,7)*{\bullet}="l3";
     (0,14)*{\bullet}="l2"; (0,21)*{\bullet}="l1";
     (13,0)*{\bullet}="r4"; (13,7)*{\bullet}="r3";
     (13,14)*{\bullet}="r2";(13,21)*{\bullet}="r1";
     "l1" ; "r1" **\dir{-}; "l2" ; "r2" **\dir{-};
     "l3" ; "r3" **\dir{-}; "l4" ; "r4" **\dir{-};
     (-4,21)*{s_1};  (-4,14)*{s_2};
     (-4,7)*{s_1'}; (-4,0)*{s_2'};
     (25,21)*{a_2 b_2 = a_2 b_2'};
     (25,14)*{a_1' b_1 = a_1' b_1' };
     (25,7)*{a_2' b_2' = a_2' b_2};
     (25,0)*{a_1 b_1' = a_1 b_1};
     "l1"; "r4" **\dir{--};
     "l4"; "r1" **\dir{--};
     "l2"; "r3" **\dir{--};
     "l3"; "r2" **\dir{--};
     \endxy};
\endxy
\]
For the first diagram, using the fact that $|A|\leq|B|\leq|S|$, it
is easy to check that imposing the extra constraint that either the
top or bottom two $S$ or $AB$ vertex labels match singles at most
$\cO(|S|^4|A||B|^3)$ terms from Eq.~(\ref{eqn:big.sum}). Similar
reasoning applies to the second diagram, but imposing the constraint
instead on rows one and four, or two and three. For all other
diagram types, at least four pairs of the indices
$a_1,a_2,a_1',a_2',b_1,b_2,b_1',b_2'$ are identical. (The number of
matching $A$ and $B$ indices is necessarily even.) In a term for
which the vertex labels are not all distinct, either a pair of $S$
indices or a further pair of $A$ or $B$ indices must be identical.
In the latter case, there must exist an identical $A$ pair
\emph{and} an identical $B$ pair among all the pairs. Again using
$|A| \leq |B| \leq |S|$, there can be at most $\cO(|S|^4|B|^3)$ such
terms per diagram type, which demonstrates the claim.

We are thus left to consider integrals with vertex label symmetry
and $N(e) \neq 0$ in Eq.~(\ref{eqn:cycle.decomp}). If $N(e) = 1$,
then our counting was correct and there is no problem. It is
therefore sufficient to bound the number of integrals in which $N(e)
> 1$. This can occur only in terms with at least 2 vertex label
symmetries. Running the argument of the previous paragraph again,
for the two diagrams with $A$ or $B$ indices all aligned, this
occurs in at most $\cO(|S|^4|B|^2)$ terms. For the rest of the
cases, it is necessary to impose equality on yet another pair of
indices, leading again to at most $\cO(|S|^4|B|^2)$ terms. Since
$\Wg(e) = \cO(|A|^{-4}|B|^{-4})$, these contributions are
collectively $\cO(|S|^2|A|^{-4}|B|^{-2})$.

The bound on the error term in Eq.~(\ref{eqn:avg.purity}) arises by
substituting the inequalities $|S| \leq |A||B|$ and $|A| \leq |B|$
into each of the estimates calculated above.

\bibliographystyle{unsrt}
\bibliography{add}

\begin{thebibliography}{10}

\bibitem{Schumacher95}
B.~Schumacher.
\newblock Quantum coding.
\newblock {\em Physical Review A}, 51:2738--2747, 1995.

\bibitem{JozsaS94}
R.~Jozsa and B.~Schumacher.
\newblock A new proof of the quantum noiseless coding theorem.
\newblock {\em J. Mod. Opt.}, 41:2343--2349, 1994.

\bibitem{P73}
J.~Pierce.
\newblock The early days of information theory.
\newblock {\em {IEEE} Transactions on Information Theory}, 19(1):3--8, 1973.

\bibitem{G64}
J.~P. Gordon.
\newblock Noise at optical frequencies; information theory.
\newblock In P.~A. Miles, editor, {\em Quantum electronics and coherent light;
  {P}roceedings of the international school of physics {E}nrico {F}ermi, Course
  {XXXI}}, pages 156--181, New York, 1964. Academic Press.

\bibitem{H73}
A.~S. Holevo.
\newblock Information theoretical aspects of quantum measurements.
\newblock {\em Probl. Info. Transm. (USSR)}, 9(2):31--42, 1973.
\newblock Translation: Probl. Info. Transm. vol. 9, pp. 177-183, 1973.

\bibitem{HausladenJSWW96}
P.~Hausladen, R.~Jozsa, B.~Schumacher, M.~Westmoreland, and W.~K. Wootters.
\newblock Classical information capacity of a quantum channel.
\newblock {\em Physical Review A}, 54:1869--1876, 1996.

\bibitem{H98}
A.~S. Holevo.
\newblock The capacity of the quantum channel with general signal states.
\newblock {\em IEEE Trans. Inf. Theory}, 44:269--273, 1998.

\bibitem{SW97}
B.~Schumacher and M.~D. Westmoreland.
\newblock Sending classical information via noisy quantum channels.
\newblock {\em Physical Review A}, 56:131--138, 1997.

\bibitem{Sh04}
P.~W. {Shor}.
\newblock Equivalence of additivity questions in quantum information theory.
\newblock {\em Communications in Mathematical Physics}, 246:453--472, 2004.
\newblock ar{X}iv:quant-ph/0305035.

\bibitem{Pomeransky03}
A.~A. {Pomeransky}.
\newblock Strong superadditivity of the entanglement of formation follows from
  its additivity.
\newblock {\em Physical Review A}, 68(3):032317--+, September 2003.
\newblock ar{X}iv:quant-ph/0305056.

\bibitem{AudenaertB04}
K.~M.~R. {Audenaert} and S.~L. {Braunstein}.
\newblock On strong superadditivity of the entanglement of formation.
\newblock {\em Communications in Mathematical Physics}, 246:443--452, 2004.
\newblock ar{X}iv:quant-ph/0303045.

\bibitem{MatsumotoSW04}
K.~{Matsumoto}, T.~{Shimono}, and A.~{Winter}.
\newblock Remarks on additivity of the {H}olevo channel capacity and of the
  entanglement of formation.
\newblock {\em Communications in Mathematical Physics}, 246:427--442, 2004.
\newblock ar{X}iv:quant-ph/0206148.

\bibitem{BennettDSW96}
C.~H. Bennett, D.~P. DiVincenzo, J.~A. Smolin, and W.~K. Wootters.
\newblock Mixed-state entanglement and quantum error correction.
\newblock {\em Physical Review A}, 54:3824--3851, 1996.
\newblock ar{X}iv:quant-ph/9604024.

\bibitem{HaydenHT01}
P.~M. {Hayden}, M.~{Horodecki}, and B.~M. {Terhal}.
\newblock The asymptotic entanglement cost of preparing a quantum state.
\newblock {\em Journal of Physics A: Mathematical and General}, 34:6891--6898,
  2001.
\newblock ar{X}iv:quant-ph/0008134.

\bibitem{VidalDC02}
G.~{Vidal}, W.~{D{\"u}r}, and J.~I. {Cirac}.
\newblock Entanglement cost of bipartite mixed states.
\newblock {\em Physical Review Letters}, 89(2):027901--+, 2002.
\newblock ar{X}iv:quant-ph/0112131.

\bibitem{MatsumotoY04}
K.~{Matsumoto} and F.~{Yura}.
\newblock Entanglement cost of antisymmetric states and additivity of capacity
  of some quantum channels.
\newblock {\em Journal of Physics A: {M}athematical and {G}eneral},
  37:L167--L171, 2004.
\newblock ar{X}iv:quant-ph/0306009.

\bibitem{VW01}
K.~G.~H. {Vollbrecht} and R.~F. {Werner}.
\newblock Entanglement measures under symmetry.
\newblock {\em Physical Review A}, 64(6):062307--+, 2001.
\newblock ar{X}iv:quant-ph/0010095.

\bibitem{KingR01}
C.~King and M.~B. Ruskai.
\newblock Minimal entropy of states emerging from noisy quantum channels.
\newblock {\em {IEEE} Transactions on Information Theory}, 47(1):192--209,
  2001.
\newblock ar{X}iv:quant-ph/9911079.

\bibitem{OsawaN01}
S.~Osawa and H.~Nagaoka.
\newblock Numerical experiments on the capacity of quantum channel with
  entangled input states.
\newblock {\em {IEICE} Transactions on Fundamentals of Electronics,
  Communications and Computer Sciences}, E84-A(10):2583--2590, 2001.
\newblock ar{X}iv:quant-ph/0007115.

\bibitem{AmosovHW00}
G.~G. Amosov, A.~S. Holevo, and R.~F. Werner.
\newblock On some additivity problems of quantum information theory.
\newblock {\em Probl. Inform. Transm.}, 36(4):25, 2000.

\bibitem{AmosovH01}
G.~G. {Amosov} and A.~S. {Holevo}.
\newblock On the multiplicativity conjecture for quantum channels.
\newblock ar{X}iv:math-ph/0103015, March 2001.

\bibitem{K01}
C.~King.
\newblock Additivity for unital qubit channels.
\newblock {\em Journal of {M}athematical {P}hysics}, 43(10):4641--4643, 2002.
\newblock ar{X}iv:quant-ph/0103156v1.

\bibitem{FujiwaraH02}
A.~{Fujiwara} and T.~{Hashizum{\'e}}.
\newblock Additivity of the capacity of depolarizing channels.
\newblock {\em Physics Letters A}, 299:469--475, July 2002.

\bibitem{K03}
C.~{King}.
\newblock The capacity of the quantum depolarizing channel.
\newblock {\em {IEEE} Transactions on Information Theory}, 49(1):221--229,
  2003.
\newblock ar{X}iv:quant-ph/0204172.

\bibitem{Holevo98a}
A.~S. Holevo.
\newblock Quantum coding theorems.
\newblock {\em Russ. Math. Surv.}, 53:1295--1331, 1998.

\bibitem{K01b}
C.~{King}.
\newblock Maximization of capacity and p-norms for some product channels.
\newblock ar{X}iv:quant-ph/0103086, 2001.

\bibitem{Sh02}
P.~W. {Shor}.
\newblock Additivity of the classical capacity of entanglement-breaking quantum
  channels.
\newblock {\em Journal of Mathematical Physics}, 43:4334--4340, 2002.
\newblock ar{X}iv:quant-ph/0201149.

\bibitem{DevetakS05}
I.~{Devetak} and P.~W. {Shor}.
\newblock The capacity of a quantum channel for simultaneous transmission of
  classical and quantum information.
\newblock {\em Communications in Mathematical Physics}, 256:287--303, June
  2005.

\bibitem{KingMNR05}
C.~{King}, K.~{Matsumoto}, M.~{Nathanson}, and M.~B. {Ruskai}.
\newblock Properties of conjugate channels with applications to additivity and
  multiplicativity.
\newblock ar{X}iv:quant-ph/0509126, 2005.

\bibitem{Cortese04}
J.~{Cortese}.
\newblock {Holevo-Schumacher-Westmoreland channel capacity for a class of qudit
  unital channels}.
\newblock {\em Physical Review A}, 69(2):022302--+, 2004.

\bibitem{DattaHS04}
N.~{Datta}, A.~S. {Holevo}, and Y.~M. {Suhov}.
\newblock A quantum channel with additive minimum output entropy.
\newblock ar{X}iv:quant-ph/0403072, 2004.

\bibitem{Fukuda05}
M.~{Fukuda}.
\newblock Extending additivity from symmetric to asymmetric channels.
\newblock {\em Journal of Physics A: Mathematical and General}, 38:L753--L758,
  November 2005.
\newblock arXiv:quant-ph/0505022.

\bibitem{Holevo04}
A.~S. Holevo.
\newblock Additivity of classical capacity and related problems.
\newblock Available online at: http://www.imaph.tu-bs.de/qi/problems/10.pdf,
  2004.

\bibitem{Holevo07}
A.~S. Holevo.
\newblock The additivity problem in quantum information theory.
\newblock In {\em Proceedings of the International Congress of Mathematicians,
  Madrid, Spain, 2006}, Publ. EMS, pages 999--1018, Zurich, 2007.

\bibitem{KingR04}
C.~King and M.~B. Ruskai.
\newblock Comments on multiplicativity of maximal {p}-norms when {p}= 2.
\newblock {\em Quantum Information and Computation}, 4:500--512, 2004.
\newblock ar{X}iv:quant-ph/0401026.

\bibitem{KingR05}
C.~King, M.~Nathanson, and M.~B. Ruskai.
\newblock Multiplicativity properties of entrywise positive maps.
\newblock {\em Linear algebra and its applications}, 404:367--379, 2005.
\newblock ar{X}iv:quant-ph/0409181.

\bibitem{SerafiniEW05}
A.~{Serafini}, J.~{Eisert}, and M.~M. {Wolf}.
\newblock Multiplicativity of maximal output purities of {G}aussian channels
  under {G}aussian inputs.
\newblock {\em \pra}, 71(1):012320--+, January 2005.

\bibitem{GiovannettiL04}
V.~Giovannetti and S.~Lloyd.
\newblock Additivity properties of a {G}aussian channel.
\newblock {\em Physical Review A}, 69:062307, 2004.
\newblock ar{X}iv:quant-ph/0403075.

\bibitem{DevetakJKR06}
I.~{Devetak}, M.~{Junge}, C.~{King}, and M.~B. {Ruskai}.
\newblock Multiplicativity of completely bounded p-norms implies a new
  additivity result.
\newblock {\em Communications in Mathematical Physics}, 266:37--63, August
  2006.

\bibitem{Michalakis07}
S.~Michalakis.
\newblock Multiplicativity of the maximal output 2-norm for depolarized
  {W}erner-{H}olevo channels.
\newblock ar{X}iv:0707.1722, 2007.

\bibitem{WH02}
R.~F. {Werner} and A.~S. {Holevo}.
\newblock Counterexample to an additivity conjecture for output purity of
  quantum channels.
\newblock {\em Journal of Mathematical Physics}, 43:4353--4357, 2002.
\newblock ar{X}iv:quant-ph/0203003.

\bibitem{AF05}
R.~Alicki and M.~Fannes.
\newblock Note on multiple additivity of minimal {R}enyi entropy output of the
  {W}erner-{H}olevo channels.
\newblock {\em Open Systems and Information Dynamics}, 11(4):339--342, 2005.
\newblock ar{X}iv:quant-ph/0407033.

\bibitem{Datta04}
N.~{Datta}.
\newblock Multiplicativity of maximal p-norms in {W}erner-{H}olevo channels for
  $1 < p < 2$.
\newblock ar{X}iv:quant-ph/0410063, 2004.

\bibitem{GiovannettiLR05}
V.~Giovannetti, S.~Lloyd, and M.~B. Ruskai.
\newblock Conditions for multiplicativity of maximal p -norms of channels for
  fixed integer p.
\newblock {\em Journal of Mathematical Physics}, 46:042105, 2005.
\newblock ar{X}iv:quant-ph/0408103.

\bibitem{W07}
A.~Winter.
\newblock The maximum output $p$-norm of quantum channels is not multiplicative
  for any $p{>}2$.
\newblock ar{X}iv:0707.0402, 2007.

\bibitem{H07}
P.~Hayden.
\newblock The maximal p-norm multiplicativity conjecture is false.
\newblock ar{X}iv.org:0707.3291, 2007.

\bibitem{HLSW04}
P.~{Hayden}, D.~{Leung}, P.~W. {Shor}, and A.~{Winter}.
\newblock {Randomizing Quantum States: Constructions and Applications}.
\newblock {\em Communications in Mathematical Physics}, 250:371--391, 2004.

\bibitem{Aubrun:no-log}
G.~Aubrun.
\newblock On almost randomizing channels with a short {K}raus decomposition.
\newblock ar{X}iv.org:0805.2900v2, 2008.

\bibitem{Paulsen86}
V.~I. Paulsen.
\newblock {\em Completely bounded maps and dilations}.
\newblock Longman Scientific and Technical, New York, 1986.

\bibitem{HLW06}
P.~{Hayden}, D.~W. {Leung}, and A.~{Winter}.
\newblock Aspects of generic entanglement.
\newblock {\em Communications in Mathematical Physics}, 265:95--117, 2006.
\newblock ar{X}iv:quant-ph/0407049.

\bibitem{BHLSW05}
C.H. Bennett, P.~Hayden, D.W. Leung, P.W. Shor, and A.~Winter.
\newblock Remote preparation of quantum states.
\newblock {\em IEEE Transactions on Information Theory}, 51(1):56--74, Jan.
  2005.
\newblock ar{X}iv:quant-ph/0307100.

\bibitem{Geman1980}
S.~Geman.
\newblock A {L}imit {T}heorem for the {N}orm of {R}andom {M}atrices.
\newblock {\em Annals of Probability}, 8(2):252--261, 1980.

\bibitem{Johnstone2001}
I.~M. Johnstone.
\newblock On the distribution of the largest eigenvalue in principal components
  analysis.
\newblock {\em Annals of Statistics}, 29(2):295--327, 2001.

\bibitem{DavidsonSzarek2001}
K.~R. Davidson and S.~J. Szarek.
\newblock Local {O}perator {T}heory, {R}andom {M}atrices and {B}anach {S}paces.
\newblock In W.~B. Johnson and J.~Lindenstrauss, editors, {\em Handbook of the
  {G}eometry of {B}anach {S}paces, {V}ol. {I}}, chapter~8, pages 317--366.
  Elsevier, 2001.

\bibitem{L01}
M.~Ledoux.
\newblock {\em The concentration of measure phenomenon}, volume~89 of {\em
  Mathematical Surveys and Monographs}.
\newblock American Mathematical Society, 2001.

\bibitem{King03b}
C.~King.
\newblock Maximal $p$-norms of entanglement breaking channels.
\newblock {\em Quantum Information and Computation}, 3(2):186--190, 2003.
\newblock ar{X}iv:quant-ph/0212057.

\bibitem{WolfE05}
M.~M. {Wolf} and J.~{Eisert}.
\newblock Classical information capacity of a class of quantum channels.
\newblock {\em New Journal of Physics}, 7:93--+, 2005.
\newblock arXiv:quant-ph/0412133.

\bibitem{p-close-to-0}
T.~Cubitt, A.~W. Harrow, D.~Leung, A.~Montanaro, and A.~Winter.
\newblock Counterexamples to additivity of minimum output p-{R}\'{e}nyi entropy
  for p close to 0.
\newblock ar{X}iv.org:0712.3628v2, 2007.

\bibitem{AmbainisSmith}
A.~Ambainis and A.~Smith.
\newblock Small pseudo-random families of matrices: Derandomizing approximate
  quantum encryption.
\newblock In {\em Proc. RANDOM}, LNCS 3122, pages 249--260. Springer, 2004.
\newblock ar{X}iv.org:quant-ph/0404075.

\bibitem{q-expanders}
A.~Ben-{A}roya and A.~Ta-{S}hma.
\newblock Quantum expanders and the quantum entropy difference problem.
\newblock ar{X}iv.org:quant-ph/0702129, 2007.

\bibitem{Hastings}
M.~B. {Hastings}.
\newblock {Random unitaries give quantum expanders}.
\newblock {\em Physical Review A}, 76(3):032315--+, 2007.
\newblock arXiv:0706.0556.

\bibitem{PWPVJ08}
D.~{P{\'e}rez-Garc{\'{\i}}a}, M.~M. {Wolf}, C.~{Palazuelos}, I.~{Villanueva},
  and M.~{Junge}.
\newblock {Unbounded Violation of Tripartite Bell Inequalities}.
\newblock {\em Communications in Mathematical Physics}, 279(2):455--486, 2008.
\newblock arXiv:quant-ph/0702189.

\bibitem{AL03}
S.~{Aubert} and C.~S. {Lam}.
\newblock Invariant integration over the unitary group.
\newblock {\em Journal of Mathematical Physics}, 44:6112--6131, 2003.
\newblock ar{X}iv:math-ph/0307012.

\bibitem{AL04}
S.~{Aubert} and C.~S. {Lam}.
\newblock Invariant and group theoretical integrations over the {U(n)} group.
\newblock {\em Journal of Mathematical Physics}, 45:3019--3039, 2004.
\newblock ar{X}iv:math-ph/0405036.

\bibitem{CS06}
B.~{Collins} and P.~{{\'S}niady}.
\newblock Integration with respect to the {H}aar measure on unitary, orthogonal
  and symplectic group.
\newblock {\em Communications in Mathematical Physics}, 264:773--795, 2006.
\newblock ar{X}iv:math-ph/0402073.

\end{thebibliography}

\end{document}